\documentclass{llncs}

\usepackage{array}
\usepackage{enumerate}
\usepackage{graphicx}
\usepackage{paralist}
\usepackage{amsmath,amssymb}
\usepackage{subfigure}
\usepackage{wrapfig}
\usepackage{url}
\usepackage{booktabs}
\usepackage{todonotes}
\usepackage[ansinew]{inputenc} %Windows
\usepackage{listings}

\newcolumntype{L}[1]{>{\raggedright\let\newline\\\arraybackslash\hspace{0pt}}m{#1}}
\newcolumntype{C}[1]{>{\centering\arraybackslash}p{#1}}
\newcolumntype{R}[1]{>{\raggedleft\let\newline\\\arraybackslash\hspace{0pt}}m{#1}}

\begin{document}

\title{Relational Algebra for In-Database Process Mining}

\author{
Remco Dijkman\inst{1} \and
Juntao Gao\inst{2} \and
Paul Grefen\inst{1} \and
Arthur ter Hofstede\inst{3,1}
}

\institute{
Eindhoven University of Technology, The Netherlands\\
\email{\{r.m.dijkman|p.w.p.j.grefen\}@tue.nl}\\
\and
Northeast Petroleum University, China\\
\email{gjt@nepu.edu.cn}\\
\and
Queensland University of Technology, Australia\\
\email{a.terhofstede@qut.edu.au}
}

\maketitle

\begin{abstract}
The execution logs that are used for process mining in practice are often obtained by querying an operational database and storing the result in a flat file. Consequently, the data processing power of the database system cannot be used anymore for this information, leading to constrained flexibility in the definition of mining patterns and limited execution performance in mining large logs. Enabling process mining directly on a database - instead of via intermediate storage in a flat file - therefore provides additional flexibility and efficiency. To help facilitate this ideal of in-database process mining, this paper formally defines a database operator that extracts the `directly follows' relation from an operational database. This operator can both be used to do in-database process mining and to flexibly evaluate process mining related queries, such as: ``which employee most frequently changes the `amount' attribute of a case from one task to the next''. We define the operator using the well-known relational algebra that forms the formal underpinning of relational databases. We formally prove equivalence properties of the operator that are useful for query optimization and present time-complexity properties of the operator. By doing so this paper formally defines the necessary relational algebraic elements of a `directly follows' operator, which are required for implementation of such an operator in a DBMS.
\end{abstract}

\keywords{process mining, relational algebra, formal methods}

\section{Introduction}

Enabling process mining directly on an operational database or data warehouse presents new opportunities. It provides additional flexibility, because event logs can be constructed on-demand by writing an SQL query, even if they are distributed over multiple tables, as is the case, for example, in SAP~\cite{ingvaldsen2007}. It even provides opportunities for fully flexible querying, allowing for the formulation of practically any process mining question. Moreover, process mining directly on a database leverages the proven technology of databases for efficiently handling large data collections in real time, which is one of the challenges identified in the process mining manifesto~\cite{aalst2012}. This can speed up process mining, especially when extremely large logs are going to be considered, such as call behavior of clients of a telecom provider, or driving behavior of cars on a road network.

\begin{table}[!tb]
	\centering
	\caption{Database table $R$ that contains a log}\label{tab:exlog}
	\begin{tabular}{
		C{0.75cm} C{1cm} C{1.75cm} C{2.0cm} p{.01cm} 
		C{0.75cm} C{1cm} C{1.75cm} C{2.0cm}}
		\cmidrule{1-4}\cmidrule{6-9}
		case & activity & start\_time & end\_time & & case & activity & start\_time & end\_time\\
		\cmidrule{1-4}\cmidrule{6-9}
		1	&	$A$	&	 00{:}20	&	 00{:}22	&	&	4	&	$A$	&	 03{:}06	 &	 03{:}10	\\
		1	&	$B$	&	 02{:}04	&	 02{:}08	&	&	4	&	$B$	&	 05{:}04	 &	 05{:}09	\\
		1	&	$E$	&	 02{:}32	&	 02{:}32	&	&	4	&	$E$	&	 07{:}26 &	 07{:}29	\\
		2	&	$A$	&	 02{:}15	&	 02{:}20&	&	5	&	$A$	&	 03{:}40 &	 03{:}44	\\
		2	&	$D$	&	 03{:}14	&	 03{:}19	&	&	5	&	$B$	&	 05{:}59 &	 06{:}06	\\
		2	&	$E$	&	 05{:}06	&	 05{:}07	&	&	5	&	$E$	&	 07{:}49 &	 07{:}52	\\
		3	&	$A$	&	 02{:}27	&	 02{:}29	&	&	6	&	$A$	&	 04{:}18 &	 04{:}20	\\
		3	&	$D$	&	 04{:}17	&	 04{:}20	&	&	6	&	$C$	&	 07{:}08 &	 07{:}12	\\
		3	&	$E$	&	 06{:}51	&	 06{:}53	&	&	6	&	$E$	&	 09{:}05 &	 09{:}07	\\
		\cmidrule{1-4}\cmidrule{6-9}
	\end{tabular}
\end{table}

To illustrate the potential benefits of process mining on a database, Table~\ref{tab:exlog} shows a very simple event log as it could be stored in a database relation. In practice, such a log would contain thousands of events, as is the case for the well know BPI Challenge logs (e.g.~\cite{dongen2011}), and even millions of events in the examples of the telecom provider and the road network mentioned above. Table~\ref{tab:exlog} shows the activities that were performed in an organization, the (customer) case for which these activities were performed and the start and end time of the activities. Given such a relation, it is important in process mining to be able to retrieve the `directly follows' relation, because this relation is the basis for many process mining techniques, including the alpha algorithm~\cite{aalst2004} (and its variants), the heuristic miner~\cite{weijters2006}, and the fuzzy miner~\cite{guenther2007}. The directly follows relation retrieves the events that follow each other directly in some case. The SQL query that retrieves this relation from relation $R$ in Table~\ref{tab:exlog} is:
\begin{verbatim}
SELECT DISTINCT a.activity, b.activity
FROM R a, R b WHERE a.case = b.case AND 
  a.end_time < b.end_time AND
  NOT EXISTS SELECT * FROM R c WHERE c.case = a.case AND 
    a.end_time < c.end_time AND c.end_time < b.end_time;
\end{verbatim}
Another example query, is the query that returns the average waiting time before each activity:
\begin{verbatim}
SELECT b.activity, average(b.start_time - a.end_time)
FROM R a, R b WHERE a.case = b.case AND 
  a.end_time < b.end_time AND
  NOT EXISTS SELECT * FROM R c WHERE c.case = a.case AND 
    a.end_time < c.end_time AND c.end_time < b.end_time
GROUP BY b.activity;
\end{verbatim}

These queries illustrate the challenges that arise when doing process mining directly on a database:
\begin{compactenum}
\item The queries are inconvenient. Even a conceptually simple process mining request like `retrieve all directly follows relations between events' is difficult to phrase in SQL.
\item The queries are inefficient. The reason for this is that the `directly follows' relation that is at the heart of process mining requires a nested query (i.e. the `NOT EXISTS' part) and nested queries are known to cause performance problems in database, irrespective of the quality of query optimization~\cite{kim1982,chaudhuri1998}, This will be discussed in detail in Section~\ref{sec:evaluation}.
\end{compactenum}

\begin{figure}
\centering
\includegraphics[width=\textwidth]{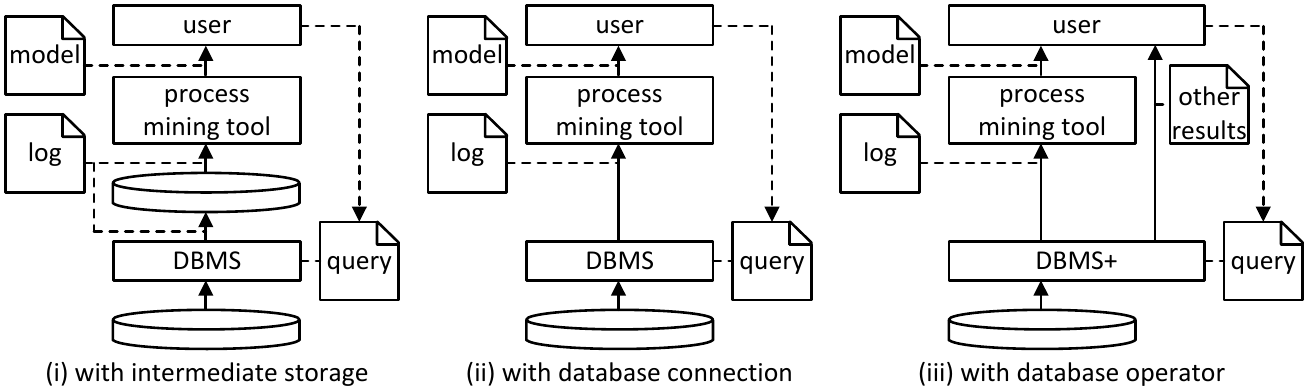}
\caption{Strategies for process mining on a database}
\label{fig:strategies}
\end{figure}

Consequently, measures must be taken to make process mining - and in particular extracting the `directly follows' relation - feasible on relational databases. Figure~\ref{fig:strategies} shows three possible strategies. Figure~\ref{fig:strategies}.i shows the current state of the art, in which a user constructs an SQL query to extract a log from the database. This log is written to disk (for example as a csv file) and then read again by a process mining tool for further processing. Consequently, the complete log must be read or written three times and there is some manual work involved in getting the log to and from disk. It is easy to imagine a process mining tool that does not need intermediate storage to disk. Such a tool would only need to read the log once and would not require manual intervention to get the log to and from disk. Figure~\ref{fig:strategies}.ii illustrates this strategy.

This paper proposes a third strategy, in which the DBMS supports a native `directly follows' operator. This strategy has the benefit that it does not require intermediate storage on disk and that it facilitates flexible and convenient querying for process mining related information. In addition, it has the benefit that it leverages proven database technology, which can efficiently handle very large data sets in real time. To realize this strategy, the paper defines the `directly follows' operator in relational algebraic form, by defining what it does, how it behaves with respect to other operators, and what its execution costs are. By doing so this paper formally defines the necessary relational algebraic elements of a `directly follows' operator, which are required for implementation of such an operator in a DBMS.

Against this background the remainder of this paper is structured as follows. Section~\ref{sec:background} explains relational algebra as background for this paper. Section~\ref{sec:algebraformining} presents a relational algebraic `directly follows' operator. Section~\ref{sec:evaluation} shows the computational cost of executing this operator and the potential effects of query optimization with respect to this operator. Finally, section~\ref{sec:relwork} presents related work and section~\ref{sec:conclusions} the conclusions.

\section{Background}\label{sec:background}

Relational algebra is at the core of every relational database system. It is used to define the execution semantics of an SQL query and to define equivalence rules to determine which execution semantics (among a set of equivalent ones) is the most efficient to execute. Before we define a relational algebraic `directly follows' operator, we provide background on these two topics.

\subsection{Relational Algebra}

In this section we briefly define the basic relational algebraic operators. We refer to the many textbooks on databases (e.g.~\cite{abiteboul1995}) for more detailed definitions.

\begin{definition}[Attribute, Schema, Relation]\label{def:rel}
An attribute is a combination of a name and a domain of possible values. A schema is a set of attributes. Any two elements $s_1, s_2$ of schema $s$ with $s_1 \neq s_2$ have different names. A relation is a combination of a schema and a set of tuples. Each tuple in a relation maps attribute names from the schema of the relation to values from the corresponding domain.
\end{definition}

For example, let $C$ be the domain of case identifiers, $E$ be the domain of activities, and $T$ be the time domain. The relation of Table~\ref{tab:exlog} has the schema $\{case: C, activity: E, start\_time: T, end\_time: T\}$ and the set of tuples $\{\{case\mapsto 1, activity\mapsto A, start\_time\mapsto 00{:}20, end\_time\mapsto 00{:}22\}, \ldots\}$.

In the remainder of this paper, we will also refer to $R$ as the set of tuples of a relation $R$. For a relation with a schema that defines an attribute with name $a$, we will use $t_a$ to refer to the value of attribute $a$ in tuple $t$.

\begin{definition}[Relational Algebra]\label{def:relalg}
Let $R, S$ be relations with schema $r, s$. Furthermore, let $a, b$ be attribute names, and $\phi$ a condition over attributes that is constructed using conjunction ($\land$), disjunction ($\lor$), negation $(\lnot$), and binary conditions ($>,\geq,=,\neq,\leq,<$) over attributes and attribute values. We define the usual relational algebra operators:
\begin{compactitem}
\item Selection: $\sigma_\phi R = \{t | t \in R, \phi(t) \}$, where $\phi(t)$ is derived from $\phi$ by replacing each attribute name $a$ by its value in tuple $t$: $t_a$. The schema of $\sigma_\phi R$ is $r$.
\item Projection: $\pi_{a,b,\ldots} R = \{\{a\mapsto t_a,b\mapsto t_b,\ldots\} | t \in R \}$. The schema of $\pi_{a,b,\ldots} R$ is $r$ restricted to the attributes with names $a, b, \ldots$.
\item Renaming: $\rho_{a/b} R = R$. The schema of $\rho_{a/b} R$ is derived from $r$ by replacing the name of attribute $a$ by $b$. In the remainder of this paper, we will also use $\rho_x R$ to represent prefixing all attributes of $R$ with $x$.
\end{compactitem}
In addition, we define the usual set theoretic operators $R \cup S$, $R \cap S$, $R - S$. These operators have the usual set-theoretic interpretation, but they require that $R$ and $S$ have the same schema. We define the Cartesian product of $R$ with schema $\{r_1: R_1, r_2: R_2, \ldots, r_n:R_n\}$ and $S$ with schema $\{s_1:S_1, s_2:S_2, \ldots, s_n:S_n\}$ as $R \times S = \{\{r_1\mapsto t_{r_1},r_2\mapsto t_{r_2},\ldots,r_n\mapsto t_{r_n},s_1\mapsto u_{s_1}, s_2\mapsto u_{s_2},\ldots,s_m\mapsto u_{s_m}\}|t \in R, u \in S\}$. The schema of $R \times S$ is $r \cup s$.

Finally, a join operator is usually defined for the commonly used operator of joining tuples from two relations that have some property in common. The join operator is a shorthand for a combination of Cartesian product and selection: $R \bowtie_\phi S = \sigma_\phi R \times S$.
\end{definition}

Table~\ref{tab:exrel} shows examples of the selection, projection, and renaming operators, applied to the relation in Table~\ref{tab:exlog}.

\begin{table}[!tb]
\centering
\caption{Example Relational Algebra Expressions}\label{tab:exrel}
\begin{tabular}{
C{0.7cm} C{1cm} C{1.7cm} C{1.7cm} p{.1cm} 
C{0.7cm} p{.1cm} 
C{0.7cm} C{1cm} C{1.7cm} C{1.7cm}
}
\multicolumn{4}{c}{$\sigma_{case=1}R$}      &  & $\pi_{case}R$ &  & \multicolumn{4}{c}{$\rho_{case/id}R$}              \\
\cmidrule{1-4}\cmidrule{6-6}\cmidrule{8-11}
case & activity & start\_time & end\_time &  & case          &  & id     & activity    & start\_time & end\_time \\
\cmidrule{1-4}\cmidrule{6-6}\cmidrule{8-11}
1    & A     & 00{:}20   & 00{:}22      &  & 1             &  & 1        & A        & 00{:}20   & 00{:}22      \\
1    & B     & 02{:}04   & 02{:}08      &  & 2             &  & 1        & B        & 02{:}04   & 02{:}08      \\
1    & E     & 02{:}32   & 02{:}32      &  & 3             &  & 1        & E        & 02{:}32   & 02{:}32      \\
\cmidrule{1-4}
     &       &             &                &  & $\ldots$      &  & $\ldots$ & $\ldots$ & $\ldots$    & $\ldots$  \\    
\cmidrule{6-6}\cmidrule{8-11}
\end{tabular}
\end{table}

\subsection{Query Optimization}\label{ssec:queryoptimization}

There are a large number of proven relational algebraic equivalences that can be used to rewrite relational algebraic equations~\cite{abiteboul1995,sagiv1978}. In the remainder of this paper, we use the following ones. Let $R, S$ be tables, $a, b, c$ be attributes, $x, y$ be attribute values, $\phi, \psi$ be conditions, and $\theta$ be a binary condition ($>,\geq,=,\neq,\leq,<$). Then:
\begin{align}
\sigma_{\phi \land \psi} R & = \sigma_{\phi} (\sigma_\psi R) \label{eqn:cascades}\\
\sigma_{\phi} (\sigma_\psi R) & = \sigma_{\psi} (\sigma_\phi R) \label{eqn:comms}\\
R \bowtie_\phi S & = S \bowtie_\phi R \label{eqn:commj}\\
(R \bowtie_\phi S) \bowtie_\psi T & = R \bowtie_\phi (S \bowtie_\psi T) \label{eqn:assj}\\
\sigma_\psi (R \bowtie_\phi S) & = (\sigma_\psi R) \bowtie_\phi S \text{, if } \psi \text{ only has attributes from } R \label{eqn:distrsj}\\
\sigma_\psi (R - S) & = (\sigma_\psi R) - (\sigma_\psi S) \label{eqn:distrsm}\\
\sigma_{a \theta x} (\rho_{b/a} R) & = \rho_{b/a} (\sigma_{b \theta x} R) \label{eqn:distrsr}\\
\pi_a (\rho_{b/a} R) & = \rho_{b/a} (\pi_b R) \label{eqn:distrpr}
\end{align}
\begin{align}
\pi_{a,b,\ldots} (\sigma_{\phi} R) & = \sigma_{\phi} (\pi_{a,b,\ldots} R) \text{, if } \phi \text{ only has attributes from } a, b, \ldots \label{eqn:distrps}\\
\pi_{a,b,\ldots} (S \bowtie_\phi R) & = (\pi_{a,\ldots} R) \bowtie_\phi (\pi_{b,\ldots} S)\text{, if } a, b, \ldots \text{can be split over } R, S \label{eqn:distrpj}\\
\pi_{a,b} ( \pi_{b,c} R) & = \pi_b R \label{eqn:intersectp}\\
\pi_{a,\ldots} ( \pi_{b,\ldots} R) & = \pi_{b,\ldots} ( \pi_{a,\ldots} R) \label{eqn:commp}\\
\rho_{b/a}(R \bowtie_{b\theta c} S) & = (\rho_{b/a}R) \bowtie_{a\theta c} S \text{, if } a,b \text{ only in } R \label{eqn:distrrj}\\
\pi_{Rs} R & = R \text{, if } Rs \text{ has only attributes from } R \label{eqn:onep}\\
\pi_{Rs} (R \bowtie_\phi S) & = R \text{, if } R \bowtie_\phi S \text{ includes each tuple of } R \nonumber \\
                            & \phantom{ = R,\ } \text{and } Rs \text{ has exactly the attributes from } R \label{eqn:zerop}\\
(R-T) \bowtie_{a \theta b} S & = R \bowtie_{a \theta b} S - T \bowtie_{a \theta b} S \label{eqn:distrjm}
\end{align}

In practice these equivalences are used to generate alternative formulas that lead to the same result, but represent alternative execution strategies. For example, $\sigma_\psi(\sigma_\phi R \times \sigma_\theta S)$ can be proven to be equivalent to $\sigma_{\psi \land \phi \land \theta} (R \times S)$. However, $\sigma_\psi(\sigma_\phi R \times \sigma_\theta S)$ represents the execution strategy in which we first execute the selections and then the Cartesian product, while $\sigma_{\psi \land \phi \land \theta} (R \times S)$ represents the execution strategy where we first execute the Cartesian product and then the selection. The first execution strategy is much more efficient than the second, because it only requires the Cartesian product to be computed for a (small) subset of $R$ and $S$.

\section{Relational Algebra for Process Mining}\label{sec:algebraformining}

This section defines the `directly follows' relation as a relational algebraic operator. It also presents and proves equivalences for this operator that can be used for query optimization, similar to the equivalences that are presented in Section~\ref{ssec:queryoptimization}.

\subsection{Directly Follows Operator}

\begin{table}[!tb]
	\centering
	\caption{Result of $>_{case, end\_time} Log$}\label{tab:exwf}
\begin{tabular}{
	C{0.8cm} C{1.3cm} C{1.7cm} C{1.7cm} C{0.8cm} C{1.3cm} C{1.7cm} C{1.7cm}}
	\cmidrule{1-8}
	$\downarrow$case & $\downarrow$activity & $\downarrow$start\_time & $\downarrow$end\_time & $\uparrow$case & $\uparrow$activity & $\uparrow$start\_time & $\uparrow$end\_time\\
	\cmidrule{1-8}
	1	&	$A$	&	 00{:}20	&	 00{:}22	&	1	&	$B$	&	02{:}04&	 02{:}08	\\
	1	&	$B$	&	 02{:}04	&	 02{:}08	&	1	&	$E$	&	02{:}32&	 02{:}32	\\
	2	&	$A$	&	 02{:}15	&	 02{:}20	&	2	&	$D$	&	03{:}14&	 03{:}19	\\
	2	&	$D$	&	 03{:}14	&	 03{:}19	&	2	&	$E$	&	05{:}06&	 05{:}07	\\
	3	&	$A$	&	 02{:}27	&	 02{:}29	&	3	&	$D$	&	04{:}17&	 04{:}20	\\
	3	&	$D$	&	 04{:}17	&	 04{:}20	&	3	&	$E$	&	06{:}51&	 06{:}53	\\
	4	&	$A$	&	 03{:}06	&	 03{:}10	&	4	&	$B$	&	05{:}04&	 05{:}09	\\
	4	&	$B$	&	 05{:}04	&	 05{:}09	&	4	&	$E$	&	07{:}26&	 07{:}29	\\
	5	&	$A$	&	 03{:}40	&	 03{:}44	&	5	&	$B$	&	05{:}59&	 06{:}06	\\
	5	&	$B$	&	 05{:}59	&	 06{:}06	&	5	&	$E$	&	07{:}49&	 07{:}52	\\
	6	&	$A$	&	 04{:}18	&	 04{:}20	&	6	&	$C$	&	07{:}08&	 07{:}12	\\
	6	&	$C$	&	 07{:}08	&	 07{:}12	&	6	&	$E$	&	09{:}05&	 09{:}07	\\
	\cmidrule{1-8}
\end{tabular}
\end{table}

The directly follows operator retrieves events that directly follow each other in some case. For a database table $Log$ that has a column $c$, which denotes the case identifier, and a column $t$, which denotes the completion timestamp of an event, we denote this operator as $>_{c,t} Log$. For example, applying the operator to the example log from Table~\ref{tab:exlog} (i.e. $>_{case, end\_time} Log$) returns Table~\ref{tab:exwf}. Similar to the way in which the join operator is defined in terms of other relational algebra operators, we define the `directly follows' operator in terms of the traditional relational algebra operators as follows.
\begin{definition}[Directly Follows Operator]\label{def:wfop}
\begin{align*}
>_{c,t} Log = &
\rho_{\downarrow} Log \bowtie_{\downarrow t < \uparrow t \land \downarrow c = \uparrow c} \rho_{\uparrow} Log \\
& - \pi_{As} ((\rho_{\downarrow} Log \bowtie_{\downarrow t < \uparrow t \land \downarrow c = \uparrow c} \rho_{\uparrow} Log) \bowtie_{\downarrow t < t \land t < \uparrow t \land \downarrow c = c} Log)
\end{align*}
where $As$ is the set of attributes that are in $\rho_{\downarrow} Log$ or $\rho_{\uparrow} Log$.
\end{definition}

The directly follows operator can both be used in an algorithm for process mining that is based on it (or on `footprints' which are derived from it~\cite{aalst2011}) and for flexible querying. Some example queries include:
\begin{compactitem}
\item   The two activities that precede a rejection: \\
        $\pi_{\uparrow\uparrow activity = reject} >_{\uparrow case, \uparrow end\_time} >_{case, end\_time} Log$
\item   The activities in which the amount of a loan is changed: \\
        $\sigma_{\uparrow amount \neq \downarrow amount} >_{\uparrow case, \uparrow end\_time} Log$
\item   The resources that ever changed the amount of a loan: \\
        $\pi_{\uparrow resource} \sigma_{\uparrow amount \neq \downarrow amount} >_{\uparrow case, \uparrow end\_time} Log$
\end{compactitem}

\subsection{Directly Follows Query Optimization}\label{ssec:optimization}

To facilitate query optimization for the directly follows operator, we must define how it behaves with respect to the other operators and prove that behavior. We present this behavior as propositions along with their proofs. In each of these propositions, we use $a, b, c, t$ as attributes (where - as convention - we use $c$ to denote the case identifier attribute and $t$ to denote the time attribute), $\theta$ as a binary operator from the set $\{>,\geq,=,\neq,\leq,<\}$, and $x$ as a value.

\setcounter{proposition}{16}

The first proposition holds for case attributes and event attributes. We define case attributes as attributes that keep the same value for all events in a case, from the moment that they get a value. We define event attributes as attributes that have a value for at most one event in each case. Consequently, we can only use this proposition for optimizing queries that involve a selection on a case or event attribute. Selections on other types of attributes (including resource attributes) cannot be optimized with this proposition.

\begin{proposition}[directly follows and selection commute]\label{eqn:selectdistrgt}\\
$
>_{c,t} \sigma_{a \theta x} Log
=
\sigma_{\downarrow a \theta x \land \uparrow a \theta x} >_{c,t} Log
$, if $a$ is a case or event attribute.
\end{proposition}

\begin{proof}
\begin{align*}
& >_{c,t} \sigma_{a \theta x} Log\\
& = \text{ (definition~\ref{def:wfop})}\\
&\qquad \rho_{\downarrow}(\sigma_{a \theta x} Log) \bowtie_{\downarrow t < \uparrow t \land \downarrow c = \uparrow c} \rho_{\uparrow}(\sigma_{a \theta x} Log) \\
&\qquad - \pi_{As} \big(((\rho_{\downarrow} (\sigma_{a \theta x} Log) \bowtie_{\downarrow t < \uparrow t \land \downarrow c = \uparrow c} \rho_{\uparrow}(\sigma_{a \theta x} Log)) \bowtie_{\downarrow t < t \land t < \uparrow t \land \downarrow c = c} \sigma_{a \theta x} Log) \big) \\
& = \text{ (proposition~\ref{eqn:distrsr})}\\
&\qquad \sigma_{\downarrow a \theta x} (\rho_{\downarrow} Log) \bowtie_{\downarrow t < \uparrow t \land \downarrow c = \uparrow c} \sigma_{\uparrow a \theta x}(\rho_{\uparrow} Log) \\
&\qquad - \pi_{As} \big(((\sigma_{\downarrow a \theta x}(\rho_{\downarrow} Log) \bowtie_{\downarrow t < \uparrow t \land \downarrow c = \uparrow c} \sigma_{\uparrow a \theta x}(\rho_{\uparrow} Log)) \bowtie_{\downarrow t < t \land t < \uparrow t \land \downarrow c = c} \sigma_{a \theta x} Log) \big)\\
& = \text{ (proposition~\ref{eqn:cascades},~\ref{eqn:commj},~\ref{eqn:distrsj})}\\
&\qquad \sigma_{\downarrow a \theta x \land \uparrow a \theta x} (\rho_{\downarrow} Log \bowtie_{\downarrow t < \uparrow t \land \downarrow c = \uparrow c} \rho_{\uparrow} Log) \\
&\qquad - \pi_{As} \big( \sigma_{\downarrow a \theta x \land \uparrow a \theta x \land a \theta x} ((\rho_{\downarrow} Log \bowtie_{\downarrow t < \uparrow t \land \downarrow c = \uparrow c} \rho_{\uparrow} Log) \bowtie_{\downarrow t < t \land t < \uparrow t \land \downarrow c = c} Log) \big)\\
& = (\text{assume } {\downarrow}a \theta x \land {\uparrow}a \theta x \Rightarrow a \theta x)\\
&\qquad \sigma_{\downarrow a \theta x \land \uparrow a \theta x} (\rho_{\downarrow} Log \bowtie_{\downarrow t < \uparrow t \land \downarrow c = \uparrow c} \rho_{\uparrow} Log) \\
&\qquad - \pi_{As} \big( \sigma_{\downarrow a \theta x \land \uparrow a \theta x} ((\rho_{\downarrow} Log \bowtie_{\downarrow t < \uparrow t \land \downarrow c = \uparrow c} \rho_{\uparrow} Log) \bowtie_{\downarrow t < t \land t < \uparrow t \land \downarrow c = c} Log) \big) \\
& = \text{ (proposition~\ref{eqn:distrps})}\\
&\qquad \sigma_{\downarrow a \theta x \land \uparrow a \theta x} (\rho_{\downarrow} Log \bowtie_{\downarrow t < \uparrow t \land \downarrow c = \uparrow c} \rho_{\uparrow} Log) \\
&\qquad - \sigma_{\downarrow a \theta x \land \uparrow a \theta x} \big( \pi_{As} ((\rho_{\downarrow} Log \bowtie_{\downarrow t < \uparrow t \land \downarrow c = \uparrow c} \rho_{\uparrow} Log) \bowtie_{\downarrow t < t \land t < \uparrow t \land \downarrow c = c} Log) \big) \\
& = \text{ (proposition~\ref{eqn:distrsm})}\\
&\qquad \sigma_{\downarrow a \theta x \land \uparrow a \theta x} \big( \rho_{\downarrow} Log \bowtie_{\downarrow t < \uparrow t \land \downarrow c = \uparrow c} \rho_{\uparrow} Log \\
&\qquad - \pi_{As} ((\rho_{\downarrow} Log \bowtie_{\downarrow t < \uparrow t \land \downarrow c = \uparrow c} \rho_{\uparrow} Log) \bowtie_{\downarrow t < t \land t < \uparrow t \land \downarrow c = c} Log) \big) \\
& = \text{ (definition~\ref{def:wfop})} \\
& \sigma_{\downarrow a \theta x \land \uparrow a \theta x} >_{c,t} Log
\end{align*}
\end{proof}

Note that proposition~\ref{eqn:distrps} requires that the condition only contains attributes that are also projected (in this case $\downarrow a, \uparrow a$ must be in $As$). This condition is satisfied as per definition~\ref{def:wfop}. Also note that the proof uses an assumption, which states that if any two events in a case have the same value for an attribute, all events for that case that are between these two (in time) must also have that value (${\downarrow}a \theta x \land {\uparrow}a \theta x \Rightarrow a \theta x$). This assumption holds for case attributes and for event attributes, which are the scope of this proposition.

The next proposition is a variant of the previous one, in which there is a condition on two attributes instead of an attribute and a value.

\begin{proposition}[directly follows and selection commute 2]\label{eqn:selectdistrgt2}\\
$
>_{c,t} \sigma_{a \theta b} Log
=
\sigma_{\downarrow a \theta \downarrow b \land \downarrow a \theta \uparrow b \land \uparrow a \theta \downarrow b \land \uparrow a \theta \uparrow b} >_{c,t} Log
$, if $a,b$ are case or event attributes.
\end{proposition}

\begin{proof}
Analogous to the proof of proposition~\ref{eqn:selectdistrgt}
\end{proof}

To prove that directly follows and projection commute, we first need to prove that projection and set minus commute, because the set minus operator is an important part of the definition of the directly follows operator. However, for the general case it is not true that projection and set minus commute. A counter example is easily constructed. Let $R = \{\{a\mapsto 1,b\mapsto 2\}\}$ and $S = \{\{a\mapsto1,b\mapsto3\}\}$. For these relations it does not hold that $\pi_{a} (R - S) = (\pi_{a} R) - (\pi_{a} S)$. However, we can prove this proposition for the special case that $S$ is a subset of $R$ and $a$ uniquely identifies tuples in $R$. Since these conditions are satisfied for the directly follows operator, it is sufficient to prove the proposition under these conditions.

\begin{proposition}[projection and restricted set minus commute]\label{eqn:distrpm}\\
$\pi_{a} (R - S) = (\pi_{a} R) - (\pi_{a} S)$, 
if $S \subseteq R$ and $a$ uniquely identifies each tuple in $R$.
\end{proposition}

\begin{proof}
This equivalence is proven by observing that $S \subseteq R$ implies that a non-surjective injective function $f:S\rightarrow R$ exists that matches each tuple $s$ in $S$ to a unique tuple $r$ in $R$. The fact that $a$ uniquely identifies tuples in $R$ (and also in $S$, because $S$ is a subset of $R$) implies that $f$ is completely determined by the values of tuples in $a$, i.e., the values of attributes other than $a$ have no consequence for $f$. Therefore, projecting $R$ and $S$ onto $a$ does not change the tuple mapping defined by $f$.

Now, looking at the left side of proposition~\ref{eqn:distrpm}, calculating the projection over the difference, means removing the attributes not in $a$ from the selected tuples in $R$ that are not in the range of $f$. Looking at the right side, calculating the difference over the projections, means removing the attributes not in a from $R$ and $S$ (which does not affect $f$) and then selecting the tuples in $R$ that are not in the range of $f$. These two are equivalent.
\end{proof}

\begin{proposition}[directly follows and restricted projection commute]\label{eqn:projectdistrgt}\\
$
>_{c,t} \pi_{c,t,a} Log
=
\pi_{\downarrow c,\downarrow t,\downarrow a, \uparrow c, \uparrow t, \uparrow a} >_{c,t} Log
$
\end{proposition}

\begin{proof}
\begin{align*}
& >_{c,t} \pi_{c,t,a} Log\\
& = \text{ (definition~\ref{def:wfop})}\\
&\qquad \rho_{\downarrow}(\pi_{c,t,a} Log) \bowtie_{\downarrow t < \uparrow t \land \downarrow c = \uparrow c} \rho_{\uparrow}(\pi_{c,t,a} Log) \\
&\qquad - \pi_{As} ((\rho_{\downarrow} (\pi_{c,t,a} Log) \bowtie_{\downarrow t < \uparrow t \land \downarrow c = \uparrow c} \rho_{\uparrow}(\pi_{c,t,a} Log)) \bowtie_{\downarrow t < t \land t < \uparrow t \land \downarrow c = c} \pi_{c,t,a} Log)\\
& = \text{ (proposition~\ref{eqn:distrpr})}\\
&\qquad \pi_{\downarrow c,\downarrow t,\downarrow a} (\rho_{\downarrow} Log) \bowtie_{\downarrow t < \uparrow t \land \downarrow c = \uparrow c} \pi_{\uparrow c,\uparrow t,\uparrow a} (\rho_{\uparrow} Log) \\
&\qquad - \pi_{As} ((\pi_{\downarrow c,\downarrow t,\downarrow a} (\rho_{\downarrow} Log) \bowtie_{\downarrow t < \uparrow t \land \downarrow c = \uparrow c} \pi_{\uparrow c,\uparrow t,\uparrow a} (\rho_{\uparrow} Log)) \bowtie_{\downarrow t < t \land t < \uparrow t \land \downarrow c = c} \pi_{c,t,a} Log)\\
& = \text{ (proposition~\ref{eqn:distrpj})}\\
&\qquad \pi_{\downarrow c,\downarrow t,\downarrow a,\uparrow c,\uparrow t,\uparrow a} (\rho_{\downarrow} Log \bowtie_{\downarrow t < \uparrow t \land \downarrow c = \uparrow c} \rho_{\uparrow} Log) \\
&\qquad - \pi_{As} \big(\pi_{\downarrow c,\downarrow t,\downarrow a,\uparrow c,\uparrow t,\uparrow a,c,t,a}((\rho_{\downarrow} Log \bowtie_{\downarrow t < \uparrow t \land \downarrow c = \uparrow c} \rho_{\uparrow} Log) \bowtie_{\downarrow t < t \land t < \uparrow t \land \downarrow c = c} Log) \big) \\
& = \text{ (proposition~\ref{eqn:intersectp},~\ref{eqn:commp})}\\
&\qquad \pi_{\downarrow c,\downarrow t,\downarrow a,\uparrow c,\uparrow t,\uparrow a} (\rho_{\downarrow} Log \bowtie_{\downarrow t < \uparrow t \land \downarrow c = \uparrow c} \rho_{\uparrow} Log) \\
&\qquad - \pi_{\downarrow c,\downarrow t,\downarrow a,\uparrow c,\uparrow t,\uparrow a} \big( \pi_{As} ((\rho_{\downarrow} Log \bowtie_{\downarrow t < \uparrow t \land \downarrow c = \uparrow c} \rho_{\uparrow} Log) \bowtie_{\downarrow t < t \land t < \uparrow t \land \downarrow c = c} Log) \big) \\
\end{align*}
\begin{align*}
& = \text{ (proposition~\ref{eqn:distrpm})}\\
&\qquad \pi_{\downarrow c,\downarrow t,\downarrow a,\uparrow c,\uparrow t,\uparrow a} \big( (\rho_{\downarrow} Log \bowtie_{\downarrow t < \uparrow t \land \downarrow c = \uparrow c} \rho_{\uparrow} Log) \\
&\qquad - \pi_{As} ((\rho_{\downarrow} Log \bowtie_{\downarrow t < \uparrow t \land \downarrow c = \uparrow c} \rho_{\uparrow} Log) \bowtie_{\downarrow t < t \land t < \uparrow t \land \downarrow c = c} Log) \big)\\
& = \text{ (definition~\ref{def:wfop})}\\
& \pi_{\downarrow c,\downarrow t,\downarrow a, \uparrow c, \uparrow t, \uparrow a} >_{c,t} Log
\end{align*}
\end{proof}

The next proposition, which states that the directly follows relation and the theta join commute, makes it explicit that the directly follows relation duplicates all attributes of a log event. Table~\ref{tab:exwf} illustrates this. However, if the case, activity and start time attribute uniquely identify an event, then there is no need to duplicate the end time attribute or any other attribute. Nonetheless, the directly follows operator adds all attributes both on the $\uparrow$ and on the $\downarrow$ side of the table. This redundancy can easily be fixed later on with a project operator and in future work additional efficiency may be achieved by avoiding this redundancy altogether.

\begin{proposition}[directly follows and theta join commute]\label{eqn:joindistrgt}\\
$
>_{c,t} (R \bowtie_{a \theta b} S)
=
(>_{c,t} R) \bowtie_{\downarrow a \theta b} S \bowtie_{\uparrow a \theta b} S
$, if each tuple from $R$ is combined with a tuple in $S$.
\end{proposition}

\begin{proof}
\begin{align*}
& >_{c,t} (R \bowtie_{a \theta b} S)\\
& = \text{ (definition~\ref{def:wfop})}\\
&\qquad \rho_{\downarrow}(R \bowtie_{a \theta b} S) \bowtie_{\downarrow t < \uparrow t \land \downarrow c = \uparrow c} \rho_{\uparrow}(R \bowtie_{a \theta b} S) \\
&\qquad - \pi_{As} ((\rho_{\downarrow} (R \bowtie_{a \theta b} S) \bowtie_{\downarrow t < \uparrow t \land \downarrow c = \uparrow c} \rho_{\uparrow}(R \bowtie_{a \theta b} S)) \bowtie_{\downarrow t < t \land t < \uparrow t \land \downarrow c = c} (R \bowtie_{a \theta b} S))\\
& = \text{ (proposition~\ref{eqn:distrrj})}\\
&\qquad (\rho_{\downarrow}R) \bowtie_{\downarrow a \theta b} S \bowtie_{\downarrow t < \uparrow t \land \downarrow c = \uparrow c} (\rho_{\uparrow}R) \bowtie_{\uparrow a \theta b} S \\
&\qquad - \pi_{As} (((\rho_{\downarrow}R) \bowtie_{\downarrow a \theta b} S) \bowtie_{\downarrow t < \uparrow t \land \downarrow c = \uparrow c} (\rho_{\uparrow}R) \bowtie_{\uparrow a \theta b} S) \bowtie_{\downarrow t < t \land t < \uparrow t \land \downarrow c = c} R \bowtie_{a \theta b} S)\\
& = \text{ (proposition~\ref{eqn:assj},~\ref{eqn:distrpj},~\ref{eqn:onep},~\ref{eqn:zerop})}\\
&\qquad \rho_{\downarrow}R \bowtie_{\downarrow t < \uparrow t \land \downarrow c = \uparrow c} \rho_{\uparrow}R \bowtie_{\downarrow a \theta b} S \bowtie_{\uparrow a \theta b} S \\
&\qquad - \pi_{As} (\rho_{\downarrow}R \bowtie_{\downarrow t < \uparrow t \land \downarrow c = \uparrow c} \rho_{\uparrow}R \bowtie_{\downarrow t < t \land t < \uparrow t \land \downarrow c = c} R) \bowtie_{\downarrow a \theta b} S \bowtie_{\uparrow a \theta b} S\\
& = \text{ (proposition~\ref{eqn:distrjm})}\\
&\qquad \big(\rho_{\downarrow}R \bowtie_{\downarrow t < \uparrow t \land \downarrow c = \uparrow c} \rho_{\uparrow}R\\
&\qquad - \pi_{As} (\rho_{\downarrow}R \bowtie_{\downarrow t < \uparrow t \land \downarrow c = \uparrow c} \rho_{\uparrow}R \bowtie_{\downarrow t < t \land t < \uparrow t \land \downarrow c = c} R)\big) \bowtie_{\downarrow a \theta b} S \bowtie_{\uparrow a \theta b} S\\
& = \text{ (proposition~\ref{eqn:distrjm})}\\
& (>_{c,t} R) \bowtie_{\downarrow a \theta b} S \bowtie_{\uparrow a \theta b} S
\end{align*}
\end{proof}

\section{Execution cost}\label{sec:evaluation}

We determine the computational cost of executing the directly follows operation, either as part of a process mining tool or as an operation that is executed directly on the database. We also determine the effect of query optimization on the directly follows operator.

\subsection{Cost of computing the directly follows relation}\label{ssec:costfunction}

The execution cost of a database operation is typically defined in terms of the number of disk blocks read or written, because reading from or writing to disk are the most expensive database operations. In line with the strategies for process mining on a database that are presented in Figure~\ref{fig:strategies}, Table~\ref{tab:executioncost} shows four execution strategies with their costs. (Note that the `with database operator' strategy from Figure~\ref{fig:strategies} is split up into two alternatives.) The cost is presented as an order of magnitude, measured in terms of the number of disk blocks $B$ that must be read or written. The number of disk blocks is linear with the number of events in the log and depends on the number of bytes needed to store an event and the number of bytes per disk block. These measures assume that the complete log fits into memory.

\begin{table}[!tb]
	\centering
	\caption{Execution costs of process mining on a database}\label{tab:executioncost}
	\begin{tabular}{l C{2mm} l}
		\cmidrule{1-3}
		execution strategy	&&	order of costs (disk blocks)\\
		\cmidrule{1-3}
		\textbf{classical process mining}	&&\\
		  \ with intermediate storage				&&	$3 \cdot B$\\
		  \ with database connection			&&	$B$\\
		\textbf{with database operator}	&&\\
		  \ native operator				&&	$B$\\
		  \ composite operator			&&	$B$ up to $B^3$\\
		\cmidrule{1-3}
	\end{tabular}
\end{table}

Process mining with intermediate storage requires that the log is read and written three times: once to query the database for the log, once to store the log to disk, and once to load the log in the process mining tool. Consequently, the complexity is $3 \cdot B$. Process mining directly on a database requires that the log is read only once. Subsequent processing can be done in memory.

In many usage scenarios more flexible querying capabilities are needed, which can benefit from access to all SQL operators. For such usage scenarios, the `directly follows' relation must be extracted directly from the database. It is easy to imagine how a native `directly follows' operator would work. Such an operator would read the log, then sort it by case identifier and timestamp, and finally return each pair of subsequent rows that have the same case identifier. Such an operator would have to read the log from disk only once and consequently has linear cost. For databases that do not have a native `directly follows' operator, the operator can be emulated using the composite formula from definition~\ref{def:wfop}. The drawback of this formula is that it requires that the intermediate results from both sides of the minus operator are stored, after which the minus operator can be applied. While this is no problem as long as the intermediate results fit into memory, the costs become prohibitive once the intermediate results must be written to disk.

On a practical level, this problem manifested itself, for example, for the log of the BPI 2011 challenge~\cite{dongen2011} on our desktop with an i5 processor, 8GB of internal memory and an SSD drive, using MySQL and the standard MySQL buffer size. Each attempt to perform a database query that involved a composite `directly follows' relation, needed at least 10 minutes to execute, which is prohibitive for interactive exploration of an event log.

\begin{figure}[!tb]
\centering
\includegraphics[scale=0.6]{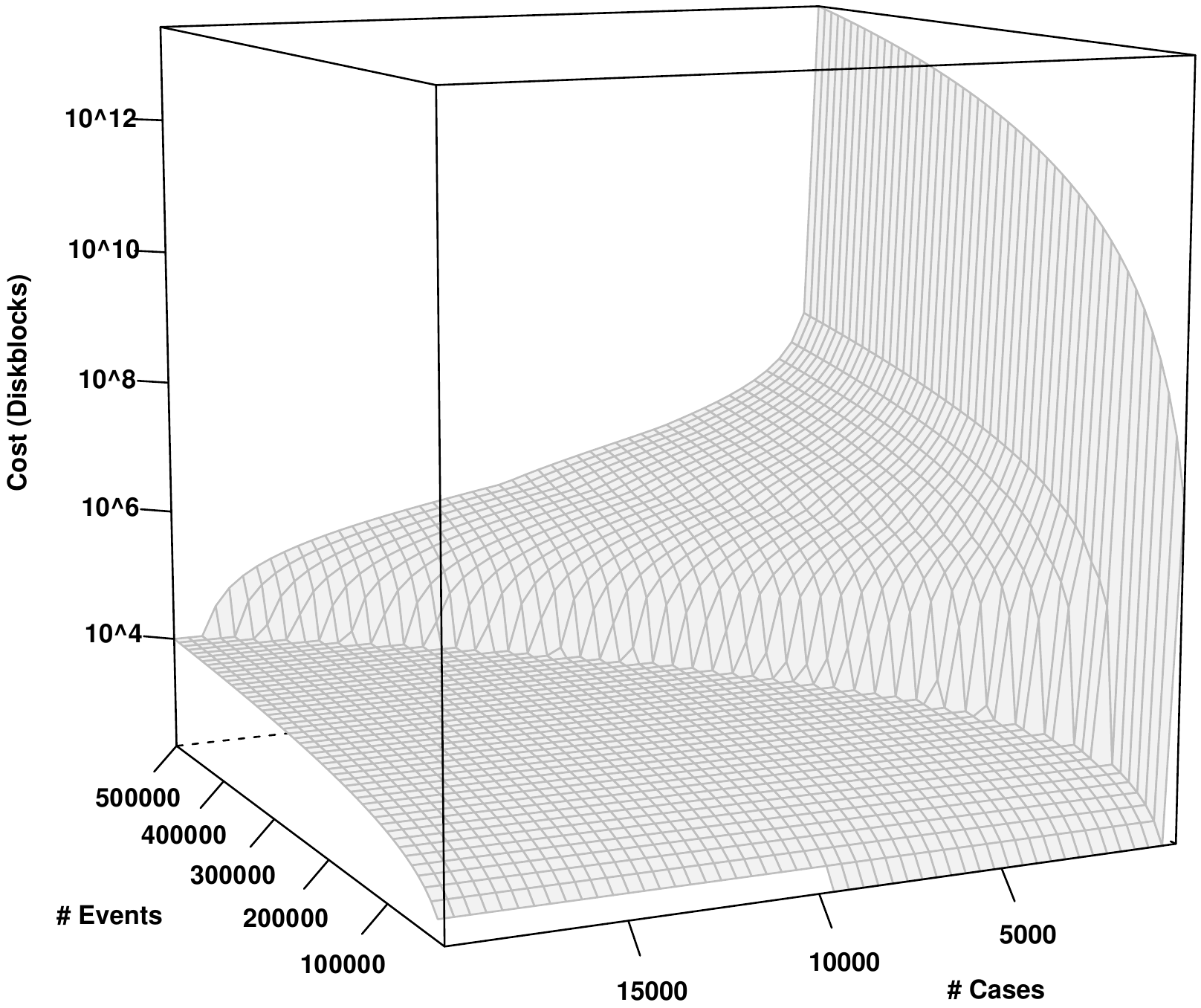}
\caption{Execution costs of process mining using the composite operator}
\label{fig:compositecost}
\end{figure}

On a theoretical level, the problem is illustrated in Figure~\ref{fig:compositecost}. This figure shows that the problem arises when the number of events in the log is high, relative to the number of cases. The mechanism that causes this is easy to derive from definition~\ref{def:wfop}, which shows that the intermediate results that must be stored are the pairs of events that directly or indirectly follow each other in some case (the left side and right side of the minus operator). Consequently, if there are many events per case, this number is high (cubic over the number of events per case in the right-hand side of the minus operator).

The precise calculation can be performed as follows. Let $V$ be the number of cases in the log, $N$ be the number of events, $F$ be the block size (i.e. the number of tuples/events that fit in a single disk block), $B_{Log} = \frac{N}{F}$ the number of disk blocks required to store the log, and $M$ be the total memory size in blocks. Note that the cost of a block nested join (or minus) operator on two relations $R$ and $S$ that take $B_R$ and $B_S$ disk blocks (with $B_R \leq B_S$), is equal to $B_R + B_S$ when one of the two relations fits in memory, and equal to $B_R + \frac{B_S}{M} \cdot B_R$ otherwise~\cite{blasgen1977}. The cost is split up into five components:
\begin{compactenum}
\item The cost of the first join is denoted as $B_{join_1}$. This equals $B_{Log}$ if the log fits into memory and $B_{Log} + \frac{B_{Log}}{M} \cdot B_{Log}$ otherwise. Note that this join appears twice, but that it only needs to be computed once.
\item The cost of storing the results of the first join to disk is denoted as $B_{result_1}$. This equals 0 if the result fits in memory. Otherwise, the number of tuples in the result, which we denote as $|t_1|$, equals the number of pairs of events that directly or indirectly follow each other in some case: $V \cdot (\frac{N}{V} \cdot \frac{N}{V}-1)/2$ on average. This fits into $\frac{|t_1|\cdot 2}{F}$ disk blocks (times 2 because each tuple in the result is a pair of tuples from the original).
\item The cost of the second join is denoted as $B_{join_2}$. This equals 0 if the original log fits into memory. Otherwise, the cost equals $B_{Log} + \frac{|t_1|\cdot 2}{F}/M \cdot B_{Log}$.
\item The cost of storing the result of the second join to disk is denoted as $B_{join_2}$. This equals 0 if the result fits into memory. Otherwise, the number of tuples in the result, which we denote as $|t_2|$, equals the number of pairs of events that indirectly follow each other. This equals the number of pairs of events $|t_1|$ that directly of indirectly follow each other minus the number of pairs of events that directly follow each other: $V\cdot (\frac{N}{V} -1)$ on average. This fits into $\frac{|t_2|\cdot 2}{F}$ disk blocks (times 2 because each tuple in the result is a triple of tuples from the original and then reduced to a pair by projection).
\item The cost of the minus operator is denoted as $B_{minus}$. This equals 0 if the result of the second join fits into memory. Otherwise, it equals $B_{result_1} + \frac{B_{result_1}}{M} \cdot B_{result_2}$.
\end{compactenum}

To generate Figure~\ref{fig:compositecost} we used a tuple size of 80 bytes, a 4 GB buffer size, and a block size of 50, such that there is a total memory size of 1 million blocks. The figure shows two `thresholds' in the computational cost. These thresholds are crossed when a particular intermediate result no longer fits into memory.

The order of the cost can be determined more easily. The order of the cost is determined the cost of the set minus, because this incorporates both intermediate results, which are typically much larger than the original log. Therefore, the order of the cost of computing the intermediate results are $\frac{N^2}{V}/F$. The total order of cost is then obtained by filling these costs out in the right-hand side of the formula for computing the cost of the set minus, which yields: $(\frac{N^2}{V}/F/M) \cdot \frac{N^2}{V}/F$. If we let $M$ be large enough to contain the log itself, but not the intermediate results (i.e. we set $M=\frac{N}{F}$), this can be simplified as: $\frac{N^3}{V^2}/F$.

Summarizing, the execution cost of flexibly retrieving a directly follows relation directly from a database can be as low as retrieving it from a process mining tool, if the database supports a native `directly follows' operator and the process mining tool supports on-database process mining. However, as long as a native `directly follows' operator does not exist, the execution costs increase to third order polynomial cost if the average number of events per case is high (i.e. if intermediate results do not fit into memory anymore).

\subsection{The effect of query optimization}

An advantage of in-database process mining is that it enables query optimization. Query optimization, using the rewrite rules that are defined in section~\ref{ssec:optimization} can greatly reduce the cost of executing a query. As an example, we show the cost of executing the query $>_{c,t} \sigma_{a \theta x} Log$ and the equivalent query $\sigma_{\downarrow a \theta x \land \uparrow a \theta x} >_{c,t} Log$. These costs decrease at least linearly with the fraction of events that match the selection criteria. Let $Q$ be that fraction. Table~\ref{tab:executionorders} shows the different execution situations that can arise. It is possible to either first derive the directly follows relation and then do the selection, or vice versa. It is also possible that the intermediate results fit in memory, or that they must be stored on disk. If the results fit in memory (and the table is indexed with respect to the variable on which the selection is done), then the execution costs are simply the cost of reading the log, or the part that matches the selection criteria, into memory once. If the intermediate results do not fit into memory, the order of the execution cost is $\frac{N^3}{V^2}/F$ as explained in the previous section. Remembering that $B = \frac{N}{F}$ leads to the formulas that are shown in the table.

\begin{table}[!ht]
	\centering
	\caption{Execution cost orders of different execution sequences}\label{tab:executionorders}
	\begin{tabular}{l C{2mm} l C{2mm} l}
		\cmidrule{1-5}
		execution sequence	&&	in memory (blocks) && on disk (blocks) \\
		\cmidrule{1-5}
		$\sigma > Log$ &&	$B$	&&	$B\cdot (\frac{N}{V})^2$\\
		$> \sigma Log$ &&	$B\cdot S$	&&	$Q\cdot B \cdot (\frac{N}{V})^2$\\
		\cmidrule{1-5}
	\end{tabular}
\end{table}

The most dramatic increase occurs when, if the selection is done first and as a consequence the intermediate results fit into memory, while if the selection is done last, the intermediate results do not fit into memory. In practice this is likely to be the case, because the selection can greatly reduce the number of events that are considered. For example, for a log with $N = 10,000$ events over $V = 500$ cases, with a block size of $F = 50$ and a selection fraction $Q = 0.10$, the order of the cost increases from $20$ to $8 \cdot 10^4$ according to the formulas from Table~\ref{tab:executionorders}. The actual computed costs increase (the same order of magnitude) from $21$ (plus one, because we need to read one disk block to load the index that is used to optimize the selection) to $9.5 \cdot 10^4$ using the formulas from the previous section.

This shows that the way in which a query that includes the directly follows operator is executed greatly influences the execution cost. Query optimizers, which are parameterized with equivalence relations from section~\ref{ssec:optimization} and the cost calculation functions from section~\ref{ssec:costfunction}, can automatically determine the optimal execution strategy for a particular query.

\section{Related Work}\label{sec:relwork}

By defining an operator for efficiently extracting the `directly follows' relation between events from a database, this paper has its basis in a long history of papers that focus on optimizing database operations. In particular, it is related to papers that focus on optimizing database operations for data mining purposes~\cite{chen1996,agrawal1993}, of which SAP HANA~\cite{farber2012} is a recent development. The idea of proposing domain-specific database operators has also been applied in other domains, such as spatio-temporal databases~\cite{abraham1999} and scientific databases~\cite{cudremauroux2009}.

By presenting a `directly follows' operator, the primary goal of this paper is to support computationally efficient process mining on a database. There exist some papers that deal with the computational complexity of the process mining algorithms themselves~\cite{maggi2012,bergenthum2007}. Also, in a research agenda for process mining the computational complexity and memory usage of process mining algorithms have been identified as important topics~\cite{aalst2004agenda}. However, this paper focuses on a step that precedes the process mining itself: flexibly querying a database to investigate which information is useful for process mining.

More database-related work from the area of process mining comes from shaping data warehouses specifically for process mining~\cite{eder2002,zurmuehlen2001}. There also exists work that focuses on the extraction of logs from a database~\cite{ingvaldsen2007}.

\section{Conclusions}\label{sec:conclusions}

This paper presents a first step towards in-database process mining. In particular, it completely defines a relational algebraic operator to extract the `directly follows' relation from a log that is stored in a relational database, possibly distributed over multiple tables. The paper presents and proves relational algebraic properties of this operator, in particular that the operator commutes with the selection, projection, and theta join. These equivalence relations can be used for query optimization. Finally, the paper presents and proves formulas to estimate the computational cost of the operator. These formulas can be used in combination with the equivalence relations to determine the most efficient execution strategy for a query. By presenting and proving these properties, the paper provides the complete formal basis for implementing the operator into a specialized DBMS for process mining, which can be used to efficiently and conveniently query a database for process mining information.

Consequently, the obvious next step of this research is to implement the operator into a DBMS. The DBMS and the relational algebraic operators can then be further extended with additional process mining-specific operators, such as an operator to query for execution traces. In addition, more algebraic properties of those operators can be proven to assist with query optimization.

There are some limitations to the equivalence relations that are presented in this paper, in particular with respect to the conditions under which they hold. These limitations restrict the possibilities for query optimization. The extent to which these theoretical limitations impact practical performance of the operator must be investigated and, if possible, mitigated.

\end{document}